\documentclass[12pt,a4paper,reqno]{article}
\usepackage{amsmath,amsthm}
\usepackage{amsmath,amssymb}
\usepackage[dvipdfm]{graphicx,color}
\usepackage[dvipdfm,
colorlinks=true,
linkcolor=blue,
citecolor=magenta,
bookmarks=true,
bookmarksnumbered=true,
bookmarkstype=toc]{hyperref}

\numberwithin{equation}{section}

\newtheorem{theorem}{Theorem}[section]
\newtheorem{lemma}[theorem]{Lemma}
\newtheorem{proposition}[theorem]{Proposition}
\newtheorem{definition}[theorem]{Definition}
\newtheorem{corollary}[theorem]{Corollary}
\newtheorem{example}[theorem]{Example}
\newtheorem{assumption}[theorem]{Assumption}
\newtheorem{remark}[theorem]{Remark}
\newtheorem{problems}[section]{Problem}

\def\lbeq(#1){\label{eqn:#1}}
\def\refeq(#1){{\rm (\ref{eqn:#1})}}
\def\lbth(#1){\label{th:#1}}
\def\lbsec(#1){\label{sec:#1}}
\def\refsec(#1){{\rm Section \ref{sec:#1}}}
\def\refsecs(#1,#2){{\rm Sections \ref{sec:#1} and \ref{sec:#2}}}
\def\refth(#1){{\rm Theorem \ref{th:#1}}}
\def\refths(#1,#2){{\rm Theorems \ref{th:#1} and \ref{th:#2}}}
\def\refthb(#1){{\bf Theorem \ref{th:#1}}}
\def\lbprp(#1){\label{prp:#1}}
\def\refprp(#1){{\rm Proposition \ref{prp:#1}}}
\def\lblm(#1){\label{lm:#1}}
\def\reflm(#1){{\rm Lemma  \ref{lm:#1}}}
\def\reflms(#1,#2){{\rm Lemmas \ref{lm:#1} and \ref{lm:#2}}}
\def\lbcor(#1){\label{cor:#1}}
\def\refcor(#1){{\rm Corollary \ref{cor:#1}}}
\def\lbrm(#1){\label{rm:#1}}
\def\refrm(#1){{\rm Remark \ref{rm:#1}}}
\def\lbexm(#1){\label{exm:#1}}
\def\refexm(#1){{\rm Example \ref{rm:#1}}}
\def\lbass(#1){\label{ass:#1}}
\def\refass(#1){{\rm Assumption \ref{ass:#1}}}
\def\refasss(#1, #2){{\rm Assumptions \ref{ass:#1} and \ref{ass:#2}}} 
\def\lbdf(#1){\label{df:#1}}
\def\refdf(#1){{\rm Definition \ref{df:#1}}}

\def\bgdf{\begin{definition}}
\def\eddf{\end{definition}}
\def\bgth{\begin{theorem}}
\def\edth{\end{theorem}}
\def\bglm{\begin{lemma}}
\def\edlm{\end{lemma}}
\def\bgprp{\begin{proposition}}
\def\edprp{\end{proposition}}
\def\bgcor{\begin{corollary}}
\def\edcor{\end{corollary}}
\def\bgexm{\begin{example}}
\def\edexm{\end{example}}
\def\bgpf{\begin{proof}}
\def\edpf{\end{proof}}
\def\bgrm{\begin{remark}}
\def\edrm{\end{remark}}
\def\bgrms{\begin{remarks}}
\def\edrms{\end{remarks}}
\def\bgass{\begin{assumption}}
\def\edass{\end{assumption}}
\def\ben{\begin{enumerate}}
\def\een{\end{enumerate}}
\def\bgeq{\begin{equation}}
\def\edeq{\end{equation}}
\def\bgds{\begin{description}}
\def\edds{\end{description}}
\def\bgpbs{\begin{problems}}
\def\edpbs{\end{problems}}

\def\ben{\begin{enumerate}}
\def\een{\end{enumerate}}
\def\bqn{\begin{equation}}
\def\eqn{\end{equation}}
\def\bqna{\begin{eqnarray}}
\def\eqna{\end{eqnarray}}
\def\bgmul{\begin{multline*}}
\def\edmul{\end{multline*}}

\def\br{\begin{array}}
\def\er{\end{array}}

\def\tL={{\tilde{L}}}

\begin{document}
\allowdisplaybreaks

\title{Absence of zero resonances of massless Dirac operators}
\author{Daisuke AIBA\footnote{
Department of Mathematics, Gakushuin University, 
1-5-1 Mejiro, Toshima-ku, Tokyo 171-8588, Japan.
E-mail: aiba@math.gakushuin.ac.jp.
Partially supported by F\^{u}j\^{u}-kai Foundation. }
}
\date{}

\allowdisplaybreaks
\maketitle

\begin{abstract} 
We consider the massless Dirac operator $ H = \alpha \cdot D + Q(x) $
on the Hilbert space $ L^{2}( \mathbb{R}^{3}, \mathbb{C}^{4} ) $,
where $ Q(x) $ is a $ 4 \times 4 $ Hermitian matrix valued function
which suitably decays at infinity.
We show that the the zero resonance is absent for $H$,
extending recent results of Sait\={o} - Umeda \cite{YT} and Zhong - Gao \cite{YG}.
\end{abstract}

\section{Introduction, assumption and theorems.}
We consider the massless Dirac operator
\begin{equation} \label{massless}
H = \alpha \cdot D + Q(x),\ D = -i \nabla_{x},\ x \in \mathbb{R}^{3},
\end{equation}
acting on $ \mathbb{C}^{4} $-valued functions on $ \mathbb{R}^{3} $.
Here $ \alpha = ( \alpha_{1}, \alpha_{2}, \alpha_{3} ) $ is the triple of $ 4 \times 4 $ Dirac matrices:
\[
 \alpha_{j} = \left(
     \begin{array}{@{\,}cccc@{\,}}
     \textbf{0} & \sigma_{j}      \\
     \sigma_{j} & \textbf{0}      
     \end{array}
\right) \ j = 1, 2, 3,
\]
with the $ 2 \times 2 $ zero matrix \textbf{0} and the triple of $ 2 \times 2 $ Pauli matrices:
\[
\sigma_{1} = \left(
     \begin{array}{@{\,}cccc@{\,}}
     0 & 1      \\
     1 & 0      
     \end{array}
\right),\ 
\sigma_{2} = \left(
     \begin{array}{@{\,}cccc@{\,}}
     0 & -i      \\
     i & 0      
     \end{array}
\right),\ 
\sigma_{3} = \left(
     \begin{array}{@{\,}cccc@{\,}}
     1 & 0     \\
     0 & -1      
     \end{array}
\right),
\]
and we use the vector notation that
$( \alpha \cdot D)u = \sum_{j=1}^{3} \alpha_{j} (-i\partial_{x_{j}})u $.
We assume that $ Q(x) $ is a $4\times4$ Hermitian matrix valued
function satisfying the following assumption:

\begin{assumption} \label{ass}
There exists positive constant $C$ and $ \rho >1 $ such that,
for each component $ q_{jk}(x)\ ( j, k = 1, \cdot \cdot \cdot, 4 ) $ of $Q(x)$,
\[
| q_{jk}(x) | \le C \langle x \rangle^{- \rho},\ \ x \in \mathbb{R}^{3}.
\]
where $ \langle x \rangle = ( 1 + |x|^{2} )^{\frac{1}{2}} $.
\end{assumption}

We remark that the Dirac operator for a Dirac particle minimally coupled
to the electromagnetic field described by the potential $ (q, A) $ is given by
\begin{equation} \label{eD}
\alpha \cdot ( D - A(x) ) + q(x)I_{4},
\end{equation}
where $ I_{4} $ is the $ 4 \times 4 $ identity matrix, and
is a special case of (\ref{ass}).
And, if $ q(x) = 0 $, (\ref{eD}) reduces to
\[
\alpha \cdot ( D - A(x) ) = \left(
     \begin{array}{@{\,}cccc@{\,}}
     \textbf{0} & \sigma \cdot ( D - A(x) )      \\
     \sigma \cdot ( D - A(x) ) & \textbf{0}      
     \end{array}
\right),
\]
where $ \sigma \cdot ( D - A(x) ) $ is the Weyl-Dirac operator.

To state the result of the paper,
we introduce some notation and terminology.
$ \mathcal{F} $ is the Fourier transform:
\[
(\mathcal{F}f)(\xi) = \frac{1}{(2\pi)^{3/2}} \int_{\mathbb{R}^{3}} f(\xi) e^{-ix \cdot \xi} d\xi.
\]
We often write $ \hat{f}(\xi) = (\mathcal{F}f)(\xi) $ and $ \check{f}(\xi) = (\mathcal{F}^{-1}f)(\xi) $.
$ \mathcal{L}^{2}(\mathbb{R}^{3}) = L^{2}( \mathbb{R}^{3}, \mathbb{C}^{4} ) $ is the Hilbert space of all $ \mathbb{C}^{4} $-valued square integrable functions.
For $ s \in \mathbb{R} $, $ \mathcal{L}^{2, s}(\mathbb{R}^{3}) = L^{2, s}( \mathbb{R}^{3}, \mathbb{C}^{4} ) := \langle x \rangle^{-s} L^{2}( \mathbb{R}^{3}, \mathbb{C}^{4} ) $ is the weighted $ \mathcal{L}^{2}(\mathbb{R}^{3}) $ space.
$ \mathcal{S}^{'}(\mathbb{R}^{3}) = \mathcal{S}^{'}( \mathbb{R}^{3}, \mathbb{C}^{4} ) $ is the space of
$ \mathbb{C}^{4} $-valued tempered distributions.
$ \mathcal{H}^{s}(\mathbb{R}^{3}) = H^{s}( \mathbb{R}^{3}, \mathbb{C}^{4} ) $
is the Sobolev space of order $s$:
\[
\mathcal{H}^{s}( \mathbb{R}^{3}) = \{ f \in \mathcal{S}^{'} (\mathbb{R}^{3}) | \hat{f} \in \mathcal{L}^{2, s} (\mathbb{R}^{3}) \}
\]
with the inner product
$ ( f, g )_{\mathcal{H}^{s}} = \sum_{j=1}^{4} ( \hat{f_{j}}, \hat{g_{j}} )_{L^{2, s}} $.
The spaces $ \mathcal{H}^{-s}(\mathbb{R}^{3}) $ and $ \mathcal{H}^{s}(\mathbb{R}^{3}) $
are duals of each other with respect to the coupling
\[
\langle f, g \rangle := \sum_{j=1}^{4}
\int_{\mathbb{R}^{3}} ( \mathcal{F}f_{j} )(\xi) \overline{ ( \mathcal{F}g_{j}(\xi) ) }d\xi,
\ \ f \in \mathcal{H}^{-s}(\mathbb{R}^{3}), g \in \mathcal{H}^{s}(\mathbb{R}^{3}).
\]
For Hilbert spaces $ \mathcal{X} $ and $ \mathcal{Y} $, $ B( \mathcal{X}, \mathcal{Y} ) $ stands for the Banach space of 
bounded operators from $ \mathcal{X} $ to $ \mathcal{Y} $,
$ B(\mathcal{X}) = B( \mathcal{X}, \mathcal{X} ) $.

It is well known that the free Dirac operator $ H_{0} := \alpha \cdot D $
is self-adjoint in $ \mathcal{L}^{2}(\mathbb{R}^{3}) $
with domain $ \mathcal{D}(H_{0}) = \mathcal{H}^{1}(\mathbb{R}^{3}) $.
Hence by the Kato-Rellich theorem,
$ H $ is also self-adjoint in $ \mathcal{L}^{2}(\mathbb{R}^{3}) $
with domain $ \mathcal{D}(H) = \mathcal{D}(H_{0}) $.
We denote their self-adjoint realizations again by $ H_{0} $ and $ H $ respectively.
In what follows,
we write $H_{0}f$ also for $(\alpha \cdot D)f$ when $f \in \mathcal{S}^{'}(\mathbb{R}^{3})$.

\begin{definition}
If $ f \in \mathcal{L}^{2}(\mathbb{R}^{3}) $ satisfies $ Hf = 0 $,
we say $ f $ is a zero mode of $ H $.
If $ f \in \mathcal{L}^{2, -3/2}(\mathbb{R}^{3}) $
satisfies $ Hf = 0 $ in the sense of distributions,
but $ f \notin \mathcal{L}^{2}(\mathbb{R}^{3}) $,
then $ f $ is said to be a zero resonance state and
zero is a resonance of it.
\end{definition}

The following is the main result of this paper.
\begin{theorem} \label{main}    
Let $Q(x)$ satisfy Assumption \ref{ass}.
Suppose $ f \in \mathcal{L}^{2, -3/2}(\mathbb{R}^{3}) $ satisfies $ Hf = 0 $
in the sense of distributions,
then for any $ \mu < 1/2 $, we have $ \langle x \rangle^{\mu}f \in \mathcal{H}^{1}(\mathbb{R}^{3}) $.
In particular, there are no resonance for $H$.
\end{theorem}

\begin{remark}
The decay result $\langle x \rangle^{\mu}f \in \mathcal{H}^{1}(\mathbb{R}^{3})$,
$\mu<1/2$ cannot be improved.
This can be seen from the example of zero mode of the Weyl-Dirac operator
which was constructed by Loss-Yau \cite{LY}.
Loss and Yau have constructed a vector potential $A_{LY}(x)$
and a zero mode $\phi_{LY}(x)$ satisfying $\sigma \cdot (D-A_{LY}(x))\phi_{LY}=0$,
where
$A_{LY}$ and $\phi_{LY}$ satisfy
$A_{LY}(x) = \mathcal{O}(\langle x \rangle^{-2})$, $|\phi_{LY}(x)| = \langle x \rangle^{-2}$.
Define $ f_{LY} = {}^{t}(0, \phi_{LY} ) $ and $Q(x) = -\alpha \cdot A_{LY}(x)$,
then
\[
Hf_{LY} = (H_{0}+Q)f_{LY}
      = \left(
     \begin{array}{@{\,}cccc@{\,}}
     \textbf{{\rm \textbf{0}}} & \sigma \cdot ( D - A_{LY}(x) )      \\
     \sigma \cdot ( D - A_{LY}(x) ) & \textbf{{\rm \textbf{0}}}      
     \end{array}
\right)f_{LY}=0,
\]
and $f_{LY} \in \mathcal{L}^{2,\mu}(\mathbb{R}^{3})$ for any $\mu < 1/2$.
However,
$f_{LY} \notin \mathcal{L}^{2, \frac{1}{2}}(\mathbb{R}^{3})$.
\end{remark}

We remark that Sait\={o} - Umeda \cite{YT} and Zhong - Gao \cite{YG} have proven
the following result under the same assumption
$ | Q(x) | \le C \langle x \rangle^{-\rho} $, $ \rho>1 $
(In \cite{YT}, it is assumed $ \rho>3/2 $,
however, arguments of \cite{YT} go through under the assumption $ \rho>1 $
as was made explicit in \cite{YG}):
If $f$ satisfies $ f \in \mathcal{L}^{2, -s}(\mathbb{R}^{3}) $ for some
$ 0<s \le {\rm min} \{ 3/2, \rho-1 \} $ and $ Hf = 0 $ in the sense of distributions,
then $ f \in \mathcal{H}^{1}(\mathbb{R}^{3}) $.
Our theorem improves over the results of \cite{YT} and \cite{YG}
by weakening the assumption $ f \in \mathcal{L}^{2, -s}(\mathbb{R}^{3}) $
to $ \mathcal{L}^{2, -3/2}(\mathbb{R}^{3}) $,
which is $ \rho >1 $ independent, and by strengthening
the result $ f \in \mathcal{H}^{1}(\mathbb{R}^{3}) $
to a sharp decay estimate $ \langle x \rangle^{\mu}f  \in \mathcal{H}^{1}(\mathbb{R}^{3}) $,
$ \mu < 1/2 $.
We briefly explain the significance of the theorem.

The solution of the time-dependent Dirac equation
\[
i \frac{\partial u}{\partial t} = Hu,\ \ u(0) = \phi
\]
is given by $ e^{-itH} \phi $.
Under Assumption \ref{ass},
it has been proven that the spectrum $ \sigma(H) = \mathbb{R} $,
the limiting absorption principle is satisfied
and that $ \sigma_{p}(H) \backslash \{0\} $ is discrete.
To make the argument simple,
we assume $ \sigma_{p}(H) \subset \{0\} $.
Then for $ \phi \in \mathcal{L}^{2}_{ac}(H) $, the absolutely continuous spectral subspace
of $ \mathcal{L}^{2} $ for $H$, $ e^{-itH}\phi $ may be represented in terms of the boundary values of the resolvent $ ( H - \lambda \pm i0 )^{-1} $:
\[
e^{-itH}\phi = \lim_{\epsilon \downarrow 0}
\frac{1}{2 \pi i} \int_{\mathbb{R} \backslash (-\epsilon, \epsilon)}
e^{-it\lambda} \{ ( H - \lambda -i0 )^{-1} - ( H - \lambda +i0 )^{-1} \} \phi d\lambda,\ t>0,
\]
and the asymptotic behavior as $ t \to \pm \infty $ of $ e^{-itH}\phi $ depends on whether
(1) $ \lambda = 0 $ is a regular point, viz, $ ( H - ( \lambda \pm i0 ) )^{-1} $
is smooth up to $ \lambda = 0 $,
(2) $ \lambda = 0 $ is a resonance of it,
(3) $ \lambda = 0 $ is an eigenvalue of $ H $ or
(4) $ \lambda = 0 $ is an eigenvalue at the same time is a resonance.
Thus, Theorem \ref{main} eliminates the possibility (2) and (4).
We should recall that if $ m \neq 0 $,
then all four cases mentioned above appear at the threshold points $ \pm m $.
It is well-known that $ \lambda = 0 $ is not a regular point
if $ f + (H_{0} \pm i0 )^{-1}Qf = 0 $ has a non-trivial solution $ f \in \mathcal{L}^{2, -\rho/2} $
and this $ f $ satisfies $ Hf = 0 $.
Note that, to conclude that $ f \in \mathcal{H}^{1} $
by using results of \cite{YT} or \cite{YG},
we need assume $ 0< \rho/2 \le \rm{min} \{ 3/2, \rho-1 \} $ or $ 2 \le \rho \le 3 $,
which is a severe restriction for this application,
whereas Theorem \ref{main} does not impose only such restriction.







The rest of the paper is devoted to the proof of Theorem \ref{main}.
In section 2, we prepare some lemmas for proving the main theorem.
In section 3, we prove the main theorem \ref{main}.


\section{Preliminaries.}  
In this section, we prepare some lemmas which are necessary for proving the theorem.
We use the following well-known lemma:
\begin{theorem}\label{NNWW}  \rm{ (Nirenberg - Walker \cite{NW}) }
Let $ 1<p< \infty $ and let $ a, b \in \mathbb{R} $ be such that $ a + b>0 $.
Define
\[
k( x, y ) = \frac{1}{ |x|^{a} |x-y|^{d-(a+b)} |y|^{b} },\ \ x, y \in \mathbb{R}^{d},\  x \neq y.
\]
Then, integral operator
\[
(K\phi)(x) = \int_{\mathbb{R}^{d}} k( x, y ) \phi(y) dy
\]
is bounded in $ L^{p}(\mathbb{R}^{d}) $ if and only if
$ a< d/p $ and $ b<d/q $, where $ q = p/(p-1) $ is the dual exponent of $p$.
\end{theorem}

For $ f = {}^{t}( f_{1}, f_{2}, f_{3}, f_{4} ) $, we define the integral operator $A$ by
\[
(Af)(x) = \frac{i}{4\pi} \int_{\mathbb{R}^{3}} \frac{\alpha \cdot (x-y)}{|x-y|^{3}} f(y) dy.
\]
Since
\[
\frac{i}{4\pi} \mathcal{F}^{-1} \left(\frac{\xi}{|\xi|^{3}}\right)(x)
= \frac{1}{(2\pi)^{\frac{3}{2}}} \frac{x}{|x|^{2}},
\]
it is obvious that
\begin{equation} \label{Q}
\mathcal{F}^{-1} (Af)(x) 
= \frac{\alpha \cdot x}{|x|^{2}} (\mathcal{F}^{-1}f)(x)
= (\alpha \cdot x)^{-1} (\mathcal{F}^{-1}f)(x)
\end{equation}

\begin{lemma} \label{A}
For any $ t \in ( -\frac{3}{2}, \frac{1}{2} ) $,
$ A \in B( \mathcal{L}^{2, -t}(\mathbb{R}^{3}), \mathcal{L}^{2, -t-1}(\mathbb{R}^{3}) ) $.
\end{lemma}

\begin{proof}
The multiplication by $ \langle x \rangle^{t} $ is isomorphism
from $ \mathcal{L}^{2}(\mathbb{R}^{3}) $ onto $ \mathcal{L}^{2, -t}(\mathbb{R}^{3}) $.
It follows that $ A \in B( \mathcal{L}^{2, -t}(\mathbb{R}^{3}), \mathcal{L}^{2, -t-1}(\mathbb{R}^{3}) ) $
if and only if $ \langle x \rangle^{-t-1} A \langle x \rangle^{t} \in B( \mathcal{L}^{2}(\mathbb{R}^{3}) ) $.
The integral kernel of $ \langle x \rangle^{-t-1} A \langle x \rangle^{t} $ is bounded by
\[
\frac{1}{4\pi \langle x \rangle^{t+1} |x-y|^{2} \langle y \rangle^{-t} }.
\]
Lemma \ref{A} follows by applying Lemma \ref{NNWW} with
$ a = t+1,\ b = -t,\ d = 3,\ p = q =2 $.
\end{proof}

\begin{lemma} \label{R} 
Let $ -3/2 < s <1/2 $.
Then for any $ g \in \mathcal{L}^{2, -s}(\mathbb{R}^{3}) $
and $ \phi \in C_{0}^{\infty}(\mathbb{R}^{3} \backslash \{0\}, \mathbb{C}^{4}) $,
we have the identity;
\begin{equation} \label{QQ}
\langle \mathcal{F}^{-1}(Ag), \phi \rangle
= \langle \mathcal{F}^{-1}g, \frac{\alpha \cdot x}{|x|^{2}} \phi \rangle.
\end{equation}
\end{lemma}

\begin{proof}
We note that both $ \mathcal{F}^{-1}g \in \mathcal{H}^{-s-1}(\mathbb{R}^{3}) $
and
$ \mathcal{F}^{-1}(Ag) \in \mathcal{H}^{-s-1}(\mathbb{R}^{3}) $.
Indeed, the former is obvious by
$ g \in \mathcal{L}^{2, -s}(\mathbb{R}^{3}) \subset \mathcal{L}^{2, -s-1}(\mathbb{R}^{3}) $
and the latter follows since
$ Ag \in \mathcal{L}^{2, -s-1}(\mathbb{R}^{3}) $
by virtue of the assumption $ g \in \mathcal{L}^{2, -s}(\mathbb{R}^{3}) $, $ -\frac{3}{2}<s<\frac{1}{2} $
and Lemma \ref{A}.
Let $ \phi \in C_{0}^{\infty}( \mathbb{R}^{3} \backslash \{0\}, \mathbb{C}^{4} ) $.
Take a sequence $ g_{n} \in C_{0}^{\infty}( \mathbb{R}^{3}, \mathbb{C}^{4} ) $ such that
$ \lim_{n \to \infty} \| g_{n} - g \|_{\mathcal{L}^{2, -s}} = 0 $.
Since $A$ is continuous from $ \mathcal{L}^{2, -s}(\mathbb{R}^{3}) $
to $ \mathcal{L}^{2, -s-1}(\mathbb{R}^{3}) $
by virtue of Lemma \ref{A},
it follows that
\begin{align*}
\langle \mathcal{F}^{-1}(Ag), \phi \rangle
&= \lim_{n \to \infty} \langle \mathcal{F}^{-1}(Ag_{n}), \phi \rangle \\
&= \lim_{n \to \infty} \langle \frac{\alpha \cdot x}{|x|^{2}} \mathcal{F}^{-1}g_{n}, \phi \rangle \\
&= \lim_{n \to \infty} \langle \mathcal{F}^{-1}g_{n}, \frac{\alpha \cdot x}{|x|^{2}} \phi \rangle \\
&= \langle \mathcal{F}^{-1}g, \frac{\alpha \cdot x}{|x|^{2}} \phi \rangle.
\end{align*}

Here we used (\ref{Q}) in the second step and
that $\dfrac{\alpha \cdot x}{|x|^{2}} \phi \in C_{0}^{\infty}( \mathbb{R}^{3} \backslash \{0\}, \mathbb{C}^{4} ) $
in the final step. This completes the proof.
\end{proof}

The following is an extension of Theorem 4.1 of \cite{YT} and
plays an important role in the proof of theorem.
\begin{lemma} \label{A2}
Suppose that $ f \in \mathcal{L}^{2, -3/2}(\mathbb{R}^{3}) $
and $ H_{0}f \in \mathcal{L}^{2, -s}(\mathbb{R}^{3}) $
for some $ s \in ( -\frac{3}{2}, \frac{1}{2} ) $.
Then, $f$ satisfies $ AH_{0}f = f $.
\end{lemma}

\begin{proof}
Since $ f \in \mathcal{L}^{2, -3/2}(\mathbb{R}^{3}) $
and $ AH_{0}f \in \mathcal{L}^{-s-1}(\mathbb{R}^{3}) \subset \mathcal{L}^{2, -3/2}(\mathbb{R}^{3}) $ by virtue of Lemma \ref{A},
it follows that $ \mathcal{F}^{-1}f, \mathcal{F}^{-1}(AH_{0}f) \in \mathcal{H}^{-3/2}(\mathbb{R}^{3}) $.
Thus, it suffice to show that
\begin{equation} \label{H}
\langle \mathcal{F}^{-1}(AH_{0}f), \phi \rangle = \langle \mathcal{F}^{-1}f, \phi \rangle,\ for\ any\ \phi \in \mathcal{H}^{3/2}(\mathbb{R}^{3}).
\end{equation}
Since $ C_{0}^{\infty}( \mathbb{R}^{d} \backslash \{0\}, \mathbb{C}^{4} ) $ is dense in $ \mathcal{H}^{s}(\mathbb{R}^{d}) $ for $ s \le d/2 $,
we need only prove (\ref{H}) for $ \phi \in C_{0}^{\infty}( \mathbb{R}^{3} \backslash \{0\}, \mathbb{C}^{4} ) $.
By setting $ g = H_{0}f $ in (\ref{QQ}) and using
$ \mathcal{F}^{-1}(H_{0}f)(x) = ( \alpha \cdot x ) ( \mathcal{F}^{-1}f )(x) $
for $ f \in \mathcal{L}^{2, -3/2}(\mathbb{R}^{3}) $,
we have
\begin{align}
\langle \mathcal{F}^{-1}(AH_{0}f), \phi \rangle 
&= \langle ( \alpha \cdot x ) \mathcal{F}^{-1}f,
\frac{\alpha \cdot x}{|x|^{2}} \phi \rangle   \notag \\
&= \langle \mathcal{F}^{-1}f, \frac{(\alpha \cdot x)^{2}}{|x|^{2}} \phi \rangle
= \langle \mathcal{F}^{-1}f, \phi \rangle. \notag
\end{align}
This completes the proof.
\end{proof}


\section{Proof of Theorem \ref{main}}  
We may assume $ 1<\rho<3 $ without losing generality.
We apply well-known Agmon's bootstrap argument.
Let $ f \in \mathcal{L}^{2, -3/2}(\mathbb{R}^{3}) $ and $ Hf=0 $
in the sense of distributions.
Then $ H_{0}f = -Qf \in \mathcal{L}^{2, -\frac{3}{2} + \rho}(\mathbb{R}^{3}) $
by the assumption \ref{ass}.
Since $ -\frac{1}{2} < \rho-\frac{3}{2} <\frac{3}{2} $,
we have $ AQf \in \mathcal{L}^{2, -\frac{3}{2}+\rho-1}(\mathbb{R}^{3}) $
by virtue of Lemma \ref{A}.
Then Lemma \ref{A2} implies $ f = AH_{0}f = -AQf \in \mathcal{L}^{2, -\frac{3}{2} +\rho-1}(\mathbb{R}^{3}) $.
Thus we may repeat the argument several times and obtain
$ f \in \mathcal{L}^{2, -\frac{3}{2} +n(\rho-1)}(\mathbb{R}^{3}) $
as long as $ -\frac{3}{2} +n(\rho-1) +1 <\frac{3}{2} $.
Let $ n_{0} $ be the largest integer such that
$ -\frac{3}{2} + n_{0} ( \rho-1 ) + 1 < \frac{3}{2} $ so that
$ f \in \mathcal{L}^{2, -\frac{3}{2}+n_{0}(\rho-1)}(\mathbb{R}^{3}) $
and $ Qf \in \mathcal{L}^{2, -\frac{3}{2}+n_{0}(\rho-1)+\rho}(\mathbb{R}^{3}) $,
however $ -\frac{3}{2}+n_{0}(\rho-1)+\rho > \frac{3}{2} $.
Then for $ \mu < 1/2 $ arbitrary close to $ 1/2 $,
$ H_{0}f = -Qf \in \mathcal{L}^{2, \mu +1}(\mathbb{R}^{3}) $.
Thus, $ f \in \mathcal{L}^{2, \mu}(\mathbb{R}^{3}) $ by virtue of Lemma \ref{A} and Lemma \ref{A2}.
By differentiating, we have
\begin{align}
H_{0} \langle x \rangle^{\mu}f &= -i \mu ( \alpha \cdot x ) \langle x \rangle^{\mu -2}f
+ \langle x \rangle^{\mu} H_{0}f \notag  \\
&= -i \mu ( \alpha \cdot x ) \langle x \rangle^{\mu -2}f - \langle x \rangle^{\mu}Qf
\in \mathcal{L}^{2}(\mathbb{R}^{3}).  \notag
\end{align}
It follows that $ \mathcal{F}( \langle x \rangle^{\mu}f ) \in \mathcal{L}^{2, 1}(\mathbb{R}^{3}) $
which is equivalent to $ \langle x \rangle^{\mu}f \in \mathcal{H}^{1}(\mathbb{R}^{3}) $.
This completes the proof of Theorem \ref{main}.



\textbf{Acknowledgment}
I would like to express my sincere thanks to Professor Kenji Yajima for his unceasing encouragement and valuable advice.


\end{document}